\documentclass{article}

\usepackage{float}
\usepackage{graphicx}
\usepackage{amsmath,amssymb}
\newtheorem{theorem}{Theorem}%[section]
\newtheorem{proposition}[theorem]{Proposition}
\newtheorem{remark}{Remark}
\newcommand{\units}[1]{\;[{\tt #1}]}
\newenvironment{proof}{\textbf{Proof:}}{\hfill$\blacksquare$}

\title{\LARGE \bf
	Less Conservative Robust Reference Governors and Their Applications
 \footnote{This work was supported by Air Force Office of Scientific Research, grant number FA9550-20-1-0385; and National Science Foundation, grant number CMMI-1904394.}
}

\author{Miguel Castroviejo-Fernandez$\;^{a,\dagger}$ \and Huayi Li$\;^a$\and\; Andrés Cotorruelo$\;^b$\and Emanuele Garone$\;^b$\and Ilya Kolmanovsky$\;^a$
\thanks{{\tt\small mcastrov@umich.edu},$\;^a$University of Michigan, Ann Arbor, MI 48109 USA. $\!^b$ Université libre de Bruxelles, Brussels, 1050, Belgium }
}

\begin{document}

\maketitle
\thispagestyle{empty}
\pagestyle{empty}

\begin{abstract}   
The applications of reference governors to systems 
with unmeasured set-bounded disturbances can lead to conservative solutions.  This conservatism can be reduced by estimating the disturbance from output measurements and canceling it in the nominal control law.  In this paper, a reference governor based on such an approach is considered and  time-varying, disturbance and state estimation errors bounding sets are derived.  Consequently, the traditional implementation of a reference governor, which exploits a constraint admissible positively-invariant set of constant commands and initial states, is replaced by one which utilizes a time-dependent sequence of similar sets (which are not necessary nested). Examples are reported which include two applications to longitudinal control of aircraft that illustrate handling of elevator uncertainty and wing icing.
% \\
% \textit{Keywords:} Reference governor, Unknown input observer, Maximum output admissible set, Aerospace applications
\end{abstract}

%===============================================================================

\section{Introduction}

Reference Governors (RGs) \cite{garone2017reference}  are supervisory schemes added to pre-stabilized or stable systems in order to ensure constraint satisfaction by altering reference commands.  The basic approach is to choose, at each time instant, a reference command such that the associated constant-reference trajectory does not lead to constraint violation. At each time instant, the reference command is chosen as close as possible to the desired reference command.

Reference governor schemes can be constructed for systems with set-bounded unmeasured disturbance inputs \cite{garone2017reference}.  However, the practical use of such robust reference governor schemes may be impeded by their potential conservatism  if the disturbance set is large and made to cover extreme values of the disturbance observed in experiments.  This conservatism typically manifests itself in a small set of commands that can be safely followed, and even lack of feasible solutions in more extreme cases.

Additionally, in traditional reference and command governor formulations, full state measurement is assumed. In practice, however, this is rarely the case and only partial noisy state measurements might be available. In \cite{ANGELI2001151}, the authors consider the feasibility of command governors (CG) for discrete time linear systems subject to bounded disturbances and measurement noise. A set-membership estimator is defined and the CG is such that constrains are respected for all elements of the set of estimated sets. The estimator set is computed online at each time step. This is computationally expensive, and can quickly become prohibitive for high dimensional systems. In this case, simpler supersets can be used but with an associated, substantial increase in conservatism. Similarly, the authors of \cite{MAYNE20061217} propose a Model Predicitve Control (MPC) formulation based on a reconstructed state vector from a partial noisy state measurement with the observer dynamics assumed to be at steady state. The disturbance and noise induced error are then combined to form a single, time invariant, disturbance set. From there, the problem is handled similarly to the Robust MPC formulation in \cite{mayne2005robust}. This idea is extended to the case where the observer dynamics have  yet to reach steady state in \cite{MAYNE20092082}, where a time varying disturbance set is obtained offline and used in the Robust MPC formulation. It is worth noting that, in practice, the time-varying set might be complex to compute unless the dynamics, disturbance set and noise set are fairly simple.

In the setting of partial state measurement and bounded unmeasured disturbances, less conservative RG designs can be obtained by incorporating more information about the disturbance. In particular, the disturbance can be estimated online with a disturbance input observer (see e.g., the survey article \cite{chen2015disturbance}) based on system output measurements, while the nominal control law can account for the disturbance estimate in generating the control signal. 
In this setting, and assuming the control constraints are not too tight, the equivalent disturbance that needs to be accommodated in the
enforcement of the constraints is the state and disturbance estimation errors.  Incorporating additional information about the disturbance, e.g., bounds on the rate of change of the disturbance,  provides set bounds %on the disturbance 
on the estimation error and can help reduce conservatism in the reference governor operation.

In this paper, we highlight the above considerations for systems represented by discrete-time linear models with disturbance inputs.  In the traditional implementation of the reference governor, a positively-invariant set of constraint admissible constant commands and initial states, $\tilde{O}_\infty$, is used; this set shrinks and eventually becomes empty as the disturbance set grows \cite{kolmanovsky1998theory}.  In the setting of state and disturbance estimation, the set bounding the disturbance and state
estimation errors is {\em time-dependent,} as, e.g., over time these errors become smaller. To account for this time dependence, we replace $\tilde{O}_\infty$ with a time-dependent sequence of sets $\tilde{O}_{\infty,n}$, where $n$ designates the discrete-time instant, and we use the latter set sequence to inform the reference selection by the reference governor.  The use of a time-dependent sequence of sets satisfying periodic invariance properties has been considered in \cite{ossareh2020reference} for the reference governor design in periodic systems; in our case, time-dependent sets satisfying suitably defined invariance and safety properties are used to construct a reference governor. Three examples are used to illustrate the application of the above less conservative approach to reference governor design.

In the next section we introduce a class of systems with matched disturbance and control inputs. For it, we then derive the prediction model based on the observer that will be used in the reference governor design. In the following two sections, we discuss the construction of ellipsoidal and polyhedral observer error bounding sets and define a time-dependent sequence of constraint admissible systems that we use for implementation of the reference governor.
After presenting theoretical results, we  extend the class of systems considered to include systems in which disturbances and control inputs are not matched. The last two sections report simulation results for the examples and concluding remarks.

\textbf{Notations:}
%Let $\mathbb{S}^n_{++}$, $\mathbb{S}^n_{+}$ denote the set of  symmetric $n\times n$ positive definite and positive  semidefinite matrices respectively. 
$\mathbb R$ denotes the set of real numbers, $\mathbb Z$ the set of integers, $\mathbb Z_{\geq 0}$ and $\mathbb Z_{>0}$ are the set of positive and strictly positive integers, respectively. $I_n$ denotes the $n\times n$ identity matrix, $0_{n\times m}$ denotes a $n$ by $m$ zero matrix.
$\mathcal{B}^n = \left\{x\in \mathbb{R}^n:\; ||x|| \leq 1\right\}$, where $|| \cdot||$ denotes the  $2$-norm.
Given $a\in\mathbb R^n,\;b\in\mathbb R^m$, $(a,b) = [a^\top,b^\top]^\top$. Given two sets $A, B\subseteq \mathbb R^n$ we denote their Pontryagin difference as $A\ominus B = \{z\in\mathbb R^n:\;z+v\in A,\;\forall v\in B\}$, and their Minkowski sum as $A\oplus B = \{a+b: a\in A,\; b\in B\}$. ${\tt int } A$ denotes the interior of $A$.  Given a matrix $P\in\mathbb R^{n\times m}$, $P\succ 0$denotes that $P$ is positive definite.

\section{System with a Disturbance Input and a Nominal Controller}
Consider the following class of linear discrete-time systems
 \begin{equation}\label{equ:1}
 \begin{array}{c}
x_{n+1} = Ax_n + B(u_n+w_n), \\
y_n = C x_n,
\end{array}
\end{equation}
where $n\in\mathbb Z_{\geq 0}$ is the current time instant, $x_n\in\mathbb R^{n_x}$ is the state vector, $u_n\in\mathbb R^{n_u}$ is the control vector, $w_n\in\mathbb R^{n_w}$ is the disturbance vector, and
$y_n\in\mathbb R^{n_y}$ is the measured output vector. Note that the disturbance and control inputs are matched. 
State and control constraints are imposed as
\begin{equation}\label{equ:2}
x_n \in X, \quad u_n \in U,\;\forall n\in\mathbb Z_{\geq0},
\end{equation}
where we assume that 
$X$ and $U$ are compact and convex sets 
satisfying $0 \in {\tt int} X$, $0 \in {\tt int} U$.
The disturbance $w_n$ is assumed to be set-bounded, \begin{equation}
    w_n \in W,
    \label{eq:WConstraint}
\end{equation}
where $W$ is a known compact and convex set with $0\in{\tt int} W$. Furthermore, we let $ w_{n+1}=w_n + \tilde{w}_n, \quad \tilde{w}_n=(w_{n+1}-w_n)$
and assume that
\begin{equation}\label{equ:3}
\tilde{w}_n \in \tilde{W},
\end{equation}
where $\tilde{w}_n$ is the change in the disturbance and $\tilde{W}$ is a known compact set with $0\in{\tt int}\tilde W$. 

The situation we are interested in is when the set $\tilde{W}$ is small, representing a slowly-varying disturbance input, $w_n$, even though the set $W$ could be large.  In such a case, applying the reference governor under the assumption $w_n \in W$ and not exploiting rate-bounds on $w_n$ can lead to conservative results.  

As the disturbance is slowly-varying and matched to the control input, we estimate it with an observer and cancel it through the control input. The control law takes the form,
\begin{equation} \label{eq:input_MatchedDist}
u_n=K \hat{x}_n+\Gamma v_n - \hat{w}_n,
\end{equation}
where the gain $K$ is stabilizing, i.e., $(A+BK)$ is a Schur matrix, % (all eigenvalues are strictly inside the unit disk),
$\Gamma$ is chosen to ensure that the steady-state, $x_v$, for a given constant $v$ matches the desired value in the absence of any disturbance, 
$\hat{x}_n$ is the state estimate and $\hat{w}_n$ is the estimated value of the disturbance. We define the output estimate, $\hat y_n$, and an extended state vector and its estimate as $x^{\tt ext}_n = (x_n,\;w_n)$ and $\hat x^{\tt ext}_n = (\hat x_n,\;\hat w_n)$ respectively.

\section{Observer and Prediction Model for Reference Governor Design}
\label{sec:matchedDist}

The estimate $\hat{x}^{\tt ext}_n$ is generated by a Luenberger observer,

\begin{subequations}
\label{eq:input_MatchedDist.1}
\begin{align}
&\hat x^{\tt ext}_{n+1} =A_{\tt ext} \hat x^{\tt ext}_n +B_{\tt ext} u_n + L (y_n-\hat{y}_n), \\
&\hat{y}_n = C_{\tt ext}\hat x^{\tt ext}_n,\quad  C_{\tt ext} = \begin{bmatrix}    C & 0_{n_y\times n_w}\end{bmatrix}, \\
 &A_{\tt ext} =  \begin{bmatrix}A &  B\\0_{n_w\times n_x}&I_{n_w}\end{bmatrix},\; B_{\tt ext} = \begin{bmatrix}B\\0_{n_w\times n_u}\end{bmatrix},\;L = \begin{bmatrix}
     L_x\\L_w
 \end{bmatrix}.\nonumber
% \end{align}
\end{align}\end{subequations}
Let 
\begin{equation*}\label{eq:input_MatchedDist.5}
e_{x,n}=x_n-\hat{x}_n,\; e_{w,n}=w_n -\hat{w}_n, \; e_n = x^{\tt ext}_n - \hat x^{\tt ext}_n.
\end{equation*}
Then,
\begin{equation}\label{eq:input_MatchedDist.75}
e_{n+1}=(A_{\tt ext} - L C_{\tt ext} )e_n+\left[\begin{array}{c} 0_{n_x\times n_w} \\ I_{n_w} \end{array} \right] \tilde{w}_n=\bar{A} e_n + \bar{B} \tilde{w}_n.
\end{equation}
We assume that the observer error dynamics are stable and the matrix $\bar{A}$ is Schur.

The state and control constraints in (\ref{equ:2}) can be recast as
\begin{align*}%\label{equ:5}
&\hat{x}_n + e_{x,n} \in X,\\
\text{and }&u_n = K \hat{x}_n +\Gamma v_n - (w_n-e_{w,n}) \in U,\\%\label{eq:5.1}\\
\text{or, }& u_n = K \hat{x}_n + \Gamma v_n + e_{w,n} \in U \ominus (-W). %\label{eq:5.11}
 \end{align*}
Using (\ref{eq:input_MatchedDist}) in (\ref{eq:input_MatchedDist.1}),
we define the prediction model with constraints,
\begin{subequations}
\label{eq:predModFin}
\begin{align}\label{equ:6}
&\hat{x}_{n+1} = A_{\tt cl} \hat{x}_n + B_{\tt cl}  v_n + B_w e_n,\\
&y_{{\tt cl},n} = C_{\tt cl} \hat{x}_n 
+ D_{\tt cl }v_n
+ D_w e_n \in Y_{\tt cl},\label{equ:6.1}
\end{align}
\end{subequations} 
where $A_{\tt cl}=A+BK$ is the closed-loop system matrix,
$B_{\tt cl}=B \Gamma$ is the closed-loop input matrix, 
\begin{align*}
&B_w= L_x C_{\tt ext},\;C_{\tt cl} =
\left[\begin{array}{c} I_{n_x \times n_x} \\ K \end{array} \right],~ D_{\tt cl}=\left[\begin{array}{c}  0_{n_x \times n_v} \\ \Gamma \end{array} \right],\\
&D_w= I_{n_x+n_w},~ Y_{\tt cl} = X \times (U\ominus (-W)).%, n_e =n_x + n_w.  
\end{align*}
\begin{remark}
In certain applications, input constraints may be imposed as $u_n+w_n \in U$ in order to avoid hitting nonlinear saturation even with actuator uncertainty. In this case, $Y_{\tt cl} = X\times U$ in (\ref{equ:6.1}).
\end{remark}

\section{Observer Error Bounding}

In this section we consider two approaches to constructing a set sequence, $\{\Omega_k\}_{k=0}^\infty$, bounding the observer error.

\subsection{Ellipsoidal bounding}
Ellipsoidal bounding techniques of reachable sets are well-developed, see e.g., \cite{SHINGIN2004389}.  As we are interested in computing these bounds systematically for different $n$, we consider this ellipsoidal bounding more explicitly.

Consider the observer error dynamics given by (\ref{eq:input_MatchedDist.75}) with the bound (\ref{equ:3}).
Let $P \succ 0$ be the positive definite solution to the discrete-time Lyapunov equation,
$$\bar{A}^\top P \bar{A} - P = -I_{n_x }.$$
Such a solution exists since $\bar{A}$ is a Schur matrix. 
Furthermore, all eigenvalues of the matrix $P$ are greater than or equal to $1$ \footnote{Indeed, if $\lambda>0$ is an eigenvalue of $P$ and $v \neq 0$ is the corresponding eigenvector (both real as $P$ is symmetric),
$v^\top\bar{A}^\top P \bar{A} v - v^\top P v = -v^\top I_{n_x \times n_x}v$ implying $(\lambda - 1)\|v\|^2 = v^\top\bar{A}^\top P \bar{A} v \geq 0$ and $\lambda \geq 1$.}. Let
\begin{align*}
    &V_k= e_k^\top P e_k,
    \end{align*}
    and define 
\begin{align*}
&\mu = 1-(\frac{c-1}{c}) \frac{1}{\lambda_{\tt max}(P)}, \;c\in\mathbb R \text{ s.t. } c >1,\\
&\rho \in\mathbb R \text{ s.t. } \rho \geq \max_{\tilde{w}_k \in \tilde{W}}
\{c \|\bar{A}^\top P \bar{B}  \tilde{w}_k\|^2+
\tilde{w}_k^\top \bar{B}^\top P \bar{B} \tilde{w}_k\}.
\end{align*}
It can be shown %(see \ref{app:App1}) 
that $V_{k+1}\leq \mu V_k+\rho$, for all $k\in\mathbb Z_{\geq 0} $.
Hence,
\begin{equation*}
	V_k \leq \mu^k V_0+\sum^{k-1}_{i=0}\mu^i\rho = \mu^k V_0 + \frac{\mu^k -1}{\mu-1} \rho.
	\end{equation*}
Note that $c>1$, and $\lambda_{max}(P) \geq 1$ ensure that  $0<\mu<1$.

Suppose that the initial observer error can be a priori set-bounded by a specified compact set $E_0$, i.e. $e_0\in E_0$.  Let
 \begin{align}&V_{0,{\tt max}}=\max_{e_0 \in E_0} e_0^\top P e_0, \text{ and}\nonumber\\
&\xi_k=\max \left\{\nu_k,~\frac{\rho}{1-\mu} \right\}, \nu_k=\mu^k V_{0,{\tt max}} + \frac{1-\mu^k}{1-\mu} \rho. \label{eq:GammaDef}\end{align}
Given that $V_k\leq \xi_k$ for all $ k\in\mathbb Z_{\geq 0}$, 
\begin{equation}\label{equ:12}
e_k \in \Omega_k\triangleq\{e\in\mathbb R^{n_x+n_w}:~e^\top P e \leq \xi_k\},\forall  k\in\mathbb Z_{\geq 0}.
\end{equation}
Furthermore, $\{\xi_k\}_{k=0}^\infty$ is a non-increasing sequence as the sequence $\{\nu_k\}_{k=0}^\infty$ converges to $\frac{\rho}{1-\mu}$ and is monotonically increasing if $\nu_0= V_{0,{\tt max}}\leq\frac{\rho}{1-\mu}$ and monotonically decreasing if $\nu_0\geq\frac{\rho}{1-\mu}$. As a result, the sequence $\{\Omega_k\}_{k=0}^\infty$ is nested, i.e. $\Omega_{k+1} \subseteq \Omega_k$ for all $k\in\mathbb Z_{\geq 0}$ and 
$$ \lim_{k \to \infty} \Omega_k = \Omega_\infty=\{e:~e^\top P e \leq
\frac{\rho}{1-\mu} \},$$
in the Hausdorff norm sense.

\subsection{Polyhedral  bounding}
The ellipsoidal error bounding sequence presented above has attractive properties but it can be a conservative estimate. To mitigate this issue, we propose an other error bounding sequence.

Suppose that the sets $E_0$ (such that $e_0 \in E_0)$ and $\tilde{W}$ are symmetric sets about the origin. 
We set
\begin{equation*}
\Omega_0=E_0. 
\end{equation*}

To estimate $\Omega_k$ for $k \geq 1$, we pick a set of normalized 
directions, $L_j \in \mathbb{R}^{n_x+n_w}$, $j=1,\cdots,n_L,$ and note that
\begin{equation*}
L_j^\top e_k = L_j^\top \bar{A}^ke_0+ L_j^\top\sum_{i=0}^{k-1-i}\bar{A}^{k-1}\bar{B}\tilde{w}_i.
\end{equation*}
Define
\begin{align*}
&w_{e,{\tt max},k,j} = w_{e,{\tt max},k-1,j}+ \max_{\tilde{w} \in  \tilde{W}} L_j^\top \bar{A}^{k-1} \bar{B}
\tilde{w},\\
&y_{e,{\tt max},k,j} = \max_{e_0 \in  E_0} L_j^\top \bar{A}^k e_0+w_{e,{\tt max},k,j}, 
\end{align*}
where $w_{e,{\tt max},0,j}=0$ for all $j=1,\cdots,n_L$.
Then,
\begin{equation}
\label{eq:OmegaPoly} 
\Omega_k=\{e:~|L_j^\top e |\leq y_{e,{\tt max},k,j},~ j=1,\cdots,n_L  \}.
\end{equation}

The set $\Omega_k$ provides an outer bound of the set of error vectors, $e_k$, that are reachable at the time instant $k$ by the disturbance sequences $\tilde{w}_0,\cdots,\tilde{w}_{k-1}$ from the initial error $e_0 \in E_0$.  
However, for such $\Omega_k$ the property $\Omega_{k+1} \subseteq \Omega_k$ is not guaranteed without additional assumptions.

Note that a vertex representation of 
(\ref{eq:OmegaPoly}) can be useful in computing the Pontryagin difference \cite{kolmanovsky1998theory}. While vertices can be computed from the representation (\ref{eq:OmegaPoly}), an alternative approach is as follows.
If the vertices of a base polyhedron, 
$\{e:~|L_j^\top e| \leq l_j \}$ are
known, then defining $\Omega_k=\{e:~|L_j^\top e| \leq l_j \max_j \{ \frac{y_{e,{\tt max},k,j}}{l_j} \} \}$
leads to a polyhedral overbounding set with known vertices that are scaled vertices of the base polyhedron.
\section{Safe Sets}
Given an error bounding sequence of sets $\{\Omega_i\}_{i=0}^\infty$, the sets of reachable outputs and states of system (\ref{eq:predModFin}) at time instant $n+k\in\mathbb Z_{\geq0}$ starting from the origin at time instant $n\in\mathbb Z_{\geq0}$: $\hat x_{n} = 0$ and considering inputs $e_i \in \Omega_i$, for all $i=n,\cdots,n+k$, are respectively given by 
\begin{align*}
&\mathcal E^y(n,k) =  D_w \Omega_{n+k} \oplus C_{\tt cl}\mathcal E^x(n,k),\\
&\mathcal E^x(n,0)=\{0\},\;\mathcal E^x(n,k) = \bigoplus_{i=0}^{k-1} A_{\tt cl}^{k-1-i} B_w \Omega_{n+i},~\forall k>0.
\end{align*}

Additionally, for each $ k \in\mathbb Z_{\geq 0}$, we define the linear operators
\begin{align*}
   & H^y_k = D_{\tt cl} +C_{\tt cl}H^x_k,\;H^x_k = (I-A_{\tt cl})^{-1} (I-A_{\tt cl}^k) B_{\tt cl},\\
   &\text{as well as }
   H^y_\infty = D_{\tt cl}+C_{\tt cl} H^x_\infty,\;H^x_\infty = (I-A_{\tt cl})^{-1}B_{\tt cl},
\end{align*}
% \vspace{-1cm}
 and the set of safe constant commands and states at the time instant $n$ that result in trajectories satisfying the constraints for all future time instants:
\begin{align} \label{equ:equ20}
&O_{\infty,n} = \bigg\{\begin{bmatrix} v \\ \hat{x}
 \end{bmatrix}:H^y_k v +C_{\tt cl} A_{\tt cl}^k \hat{x}
 \in Y_{{\tt cl},n,k},\forall k \in \mathbb N\bigg\}, 
\end{align}
where $Y_{{\tt cl},n,0}=Y_{\tt cl} \ominus D_w \Omega_n$ and for $k\geq 1$,
\begin{align}
Y_{{\tt cl},n,k}&=Y_{\tt cl} \ominus D_w \Omega_{n+k} \ominus_{i = 0}^{k-1} C_{\tt cl} A_{\tt cl}^{k-1-i} B_w \Omega_{n+i}\nonumber\\
            &= Y_{\tt cl} \ominus \mathcal E^y(n,k), \label{equ:ineqj2}
\end{align}
where the second line follows from the identity $S_1\ominus(S_2\oplus S_3) = (S_1\ominus S_2)\ominus S_3 =(S_1\ominus S_3)\ominus S_2 $ \cite[Theorem 2.1]{kolmanovsky1998theory} that holds for any three sets $S_1,S_2,S_3$.
The following properties are derived from these definitions by straightforward algebraic manipulations for any $k \geq 1,~n\geq 0$:
\begin{subequations}
    \label{equ:ivk3522}
    \begin{align}
        \label{equ:ivk3522-1}
    &H^y_{k-1}+C_{\tt cl}A_{\tt cl}^{k-1}B_{\tt cl}=H_{k}^y,\\
\label{equ:ivk3522-2}
&    Y_{{\tt cl},n,k} \oplus C_{\tt cl} A_{\tt cl}^{k-1} B_w \Omega_n \subseteq Y_{{\tt cl},n+1,k-1},
 \end{align}
\end{subequations}
where to derive (\ref{equ:ivk3522-2}) we use the property $(S_1\ominus S_2) \oplus S_2 \subseteq S_1$ \cite[Theorem 2.1 ]{kolmanovsky1998theory} which holds for 
any two sets $S_1$ and $S_2$.  
Note that $O_{\infty,n}$ is formed by an infinite number of inequalities.  It often turns out that the inequalities corresponding to larger $k$ are redundant and can be eliminated, leading to an exact representation of $O_{\infty,n}$ by a finite number of inequalities \cite{gilbert1991linear}.  However, this does not always happen, in general.  Hence a finitely-determined subset of $O_{\infty,n}$ is defined by slightly tightening the range of commands.  

Specifically, suppose we are able to construct
\footnote{If $\Omega_{n+k} \subseteq \bar{\Omega}$, for all $k \geq 0$, then this construction involves outer bounding the minimum invariant set of 
$\hat{x}_{k+1}=A_{cl} \hat{x}_k+B_w e_k,~ 
e_k \in \bar{\Omega}$, which can be done with methods in \cite{ong2006minimal,rakovic2005invariant}. If $\bar{F}$ is this outer bound, then
$\tilde{Y}_n=Y_{cl} \ominus D_w \bar{\Omega} \ominus C \bar{F} \ominus \epsilon_n 
\mathcal{B}^{n_x}$.
}
nonempty sets $\tilde{Y}_n$, $n \geq 0$, such that
\begin{equation}\label{equ:Yn}
\tilde{Y}_n \subseteq \tilde{Y}_{n+1},
\end{equation}
and, for some $\epsilon_n>0$, 
\begin{equation*}
\tilde{Y}_n \oplus \epsilon_n \mathcal B^{n_y} \subseteq Y_{{\tt cl},n,k},\quad \text{for all } k\in\mathbb Z_{\geq 0}.%k=0,1,\cdots
\end{equation*}

We have  the following result \cite{kolmanovsky1995maximal}: 
\begin{proposition}
     Suppose $(C_{\tt cl}, A_{\tt cl})$ is observable, $A_{\tt cl}$ is Schur, 
     %$0 \in {\tt int} \tilde{Y}_n$ 
     and $Y_{\tt cl}$ is compact. Then there exists $k^*(n) \in \mathbb{Z}_{\geq 0}$ such that
\begin{align} \label{equ:prop11}
&\tilde{O}_{\infty,n} =  \bigg\{( v, \hat{x}):~
    H^y_\infty v \in \tilde{Y}_{n}, \nonumber\\
    & 	H^y_k v + C_{\tt cl} A_{\tt cl}^k \hat{x}
	\in Y_{{\tt cl},n,k},\;\forall k \geq 0 
	 \bigg\} \\
  &=   \bigg\{( v, \hat{x}):~
    H^y_\infty v \in \tilde{Y}_{n}, \nonumber\\
    & 	H^y_k v + C_{\tt cl} A_{\tt cl}^k \hat{x}
	\in Y_{{\tt cl},n,k},\;k = 0,1,\cdots,k^*(n) 
	 \bigg\}, \nonumber
\end{align}
i.e., $\tilde{O}_{\infty,n}$ is a finitely-determined subset of ${O}_{\infty,n}$.
\end{proposition}
To construct $\tilde{O}_{\infty,n}$, $k$ is incremented, and new inequalities in (\ref{equ:prop11}) are added until newly added inequalities become redundant  
\cite{kolmanovsky1995maximal}.

The following proposition characterizes invariance properties of the sequence of sets $\{\tilde{O}_{\infty,k}\}_{k=0}^\infty$ under constant reference commands and ensures that by maintaining a constant reference command violation of the constraints can be avoided:
\begin{proposition} \label{prop:RecFeas} 
If $( v_n, \hat{x}_n)\in \tilde{O}_{\infty,n}$ then, for all $e_n \in \Omega_n$,
 $$\begin{bmatrix} v_n \\ A_{\tt cl} \hat{x}_n + B_{\tt cl} v_n + B_w e_n
\end{bmatrix} \in \tilde{O}_{\infty,n+1}$$.
\end{proposition}
\begin{proof}  
Since $\tilde{Y}_n \subseteq \tilde{Y}_{n+1}$, 
\begin{equation*}
    H^y_0 v_n \in \tilde{Y}_n\Rightarrow H^y_0 v_n \in \tilde{Y}_{n+1}.
\end{equation*}
Recall that $( v_n, \hat{x}_n) \in \tilde{O}_{\infty,n}$ implies that
\begin{equation*} H^y_k v_n + C_{\tt cl} A_{\tt cl}^k \hat{x}_n
\in Y_{{\tt cl},n,k}\quad \text{for all }k \geq 0,
\end{equation*}
and consider the predicted system response,  after $(k-1)\geq 0$ steps, starting from the successor state $\hat x_{n+1}$ and keeping the reference command equal to $v_n$ over the prediction horizon.
Using (\ref{equ:ivk3522-1}) and (\ref{equ:ivk3522-2}), we have that
\begin{align*}
&C_{\tt cl} A_{\tt cl}^{k-1} \big(A_{\tt cl} \hat{x}_n+B_{\tt cl} v_n + B_w e_n\big) +
H^y_{k-1}v_n \\
&=C_{\tt cl}\big(A_{\tt cl}^{k}  \hat{x}_n+ A_{\tt cl}^{k-1} B_w e_n\big)+\big [C_{\tt cl} A_{\tt cl}^{k-1} B_{\tt cl} + H^y_{k-1}\big] v_n \\
&=C_{\tt cl} A_{\tt cl}^k \hat{x}_n+H^y_k v_n + C_{\tt cl} A_{\tt cl}^{k-1} B_w e_n \\
& \subseteq Y_{{\tt cl},n,k} \oplus C_{\tt cl} A_{\tt cl}^{k-1}B_w \Omega_n 
\subseteq Y_{{\tt cl},n+1,k-1}.
\end{align*}
Thus
$$\left[\begin{array}{c} v_n \\ A_{\tt cl} \hat{x}_n + B_{\tt cl} v_n + B_w e_n
\end{array} \right] \in \tilde{O}_{\infty,n+1}.$$
\end{proof}

Simplifications occur if the bounding set sequence $\{\Omega_k\}_{k=0}^\infty$
is eventually nested:

\begin{proposition} \label{prop:MOAS_cst1} %Suppose the property $\Omega_{k+1} \subseteq \Omega_k$ holds for all $k \geq \bar{n}$.
Suppose that for all $n \geq \bar{n} \geq 0$, $\Omega_{n+1} \subseteq \Omega_n$ and $\tilde{Y}_{n} \subseteq \tilde{Y}_{n+1}$.  Then 
for all $n \geq \bar n$, 
$ k \geq 0,$ it follows that
$\mathcal E^x(n+1,k) \subseteq \mathcal E^x(n,k)$, 
$
\mathcal E^y(n+1,k) \subseteq \mathcal E^y(n,k),$ 
$Y_{{\tt cl},n,k} \subseteq Y_{{\tt cl},n+1,k}$, and $\tilde{O}_{\infty,n} \subseteq \tilde{O}_{\infty,n+1}$.
Furthermore, if 
for all $n \geq \bar{n} \geq 0$, $\Omega_{n+1} = \Omega_{\bar{n}}$ and $\tilde{Y}_{n} = \tilde{Y}_{\bar{n}}$, then 
it follows that
$\tilde{O}_{\infty,n} = \tilde{O}_{\infty,\bar{n}}$ for all $n \geq \bar n$.
\end{proposition}
\begin{proof}
The proof follows directly from the definitions of the sets involved and the assumptions.
\end{proof}

 Importantly, Proposition \ref{prop:MOAS_cst1} highlights conditions under which $\{\tilde{O}_{\infty,n}\}_{n=0}^{\infty}$ can be replaced by a finite sequence: If we replace the tail of the original sequence of error bounding sets $\{\Omega_k\}_{k=\bar n}^\infty$, $\bar n\geq 0$, by a nested sequence of sets $\{\tilde \Omega_k\}_{k=\bar n}^\infty$ such that $\Omega_k\subseteq \tilde \Omega_{k}$, then the associated $\{\tilde{O}_{\infty,n}\}_{n=0}^{\infty}$ is eventually nested. Moreover, if $\{\tilde\Omega_k\}_{k=\bar n}^\infty$ is a constant sequence, i.e. there is $\Omega_{\tt f}\subset\mathbb R^{n_x+n_w}$ such that $\tilde \Omega_n=\Omega_{\tt f}$ for all $n \geq \bar{n}$, then $\{\tilde{O}_{\infty,n}\}_{n=0}^{\infty}$ contains the same information as the finite sequence $\{\tilde{O}_{\infty,n}\}_{n=0}^{\bar n}$.  

Note that the sequence of $\tilde O_{\infty,n}$ does not depend on the current value of the state or disturbance. It can be computed and stored in memory ahead of time to reduce online computational burden.
\section{Reference governor}

At the time instant $n$, the reference governor solves the following problem:
\begin{equation}
\kappa_n^* = \max_{0\leq\kappa\leq1} \kappa \mbox{ subject to }~\begin{bmatrix} v_{n-1}+\kappa (r_n-v_{n-1}) \\ \hat{x}_n\end{bmatrix} \in \tilde{O}_{\infty,n}.\label{eq:kappaOpt}
\end{equation} 
The applied reference is then determined as
\begin{equation}
v_n=v_{n-1}+\kappa^*_{n}(r_n-v_{n-1}).\label{eq:RIncrement}
\end{equation}

The following proposition shows that when using a Reference Governor, constraints are enforced at all times
\begin{proposition} \label{prop:RG_constraints} 
Consider system (\ref{equ:1}) subject to constraints (\ref{equ:2})-(\ref{equ:3}) controlled through (\ref{eq:input_MatchedDist}) with $v_n$ selected through (\ref{eq:kappaOpt})-(\ref{eq:RIncrement}). If $(v_0,\hat{x}_0)\in \tilde{O}_{\infty,0}$ then $(x_{n},u_n) \in X\times U$ for all $n \geq 0$.
\end{proposition}

\begin{proof}
From (\ref{eq:kappaOpt}) and noting that   $v_n=v_{n-1}$
    is a feasible choice per Proposition~\ref{prop:RecFeas},
the reference governor maintains $(v_n,\hat{x}_n)\in \tilde{O}_{\infty,n}$.  Conditions in (\ref{equ:prop11}) for $k=0$ imply $D_{\tt cl} v_n+C_{\tt cl} \hat{x}_n+D_w \Omega_n \subseteq Y_{\tt cl}$ and since $e_n \in \Omega_n$, it follows that $y_{{\tt cl},n}=D_{\tt cl} v_n+C_{\tt cl} \hat{x}_n+D_w e_n \in Y_{\tt cl}$ and thus $(x_{n},u_n) \in X\times U$, for all $n \geq 0$.
\end{proof}

The next proposition establishes, under suitable assumptions, the finite-time convergence of $v_n$ to $r_n$:
\begin{proposition} \label{prop:RG_convergence} 
Suppose $A_{\tt cl}$ is Schur, $(v_0,\hat{x}_0)\in \tilde{O}_{\infty,0}$, and $r_n = r$ for all $n$ sufficiently large.
Suppose, furthermore, there exists $\bar{n}>0$ such that $\Omega_n \subseteq \Omega_{\bar{n}}$ and (\ref{equ:Yn}) holds 
 for all $n \geq \bar{n}$, where $\tilde{Y}_{\bar{n}}$
is convex and $\Omega_{\bar{n}}$ is compact,
and
$H_\infty^y r \in \tilde{Y}_{\bar{n}}$. 
Suppose also that the set,
$\bar{O}_\infty=\{(v,\hat{x}):~ H_\infty^y v \in \tilde{Y}_{\bar{n}},~
H_k^y v+ C_{\tt cl} A_{\tt cl}^k \hat{x} \in Y_{\tt cl} \ominus D_w \Omega_{\bar{n}}
\ominus_{i = 0}^{\infty} C_{\tt cl} 
A_{\tt cl}^{i} B_w \Omega_{\bar{n}},~k=0,1,\cdots 
\}
$ 
is nonempty.
Then (\ref{eq:kappaOpt})-(\ref{eq:RIncrement}) ensures that $v_n=r$ for all $n$ sufficiently large.
\end{proposition}

\begin{proof}
Without loss of generality, suppose $r_n=r$ for all $n \geq 0$.
From (\ref{eq:RIncrement}), $\{v_n\}_{n=0}^\infty$ is a sequence of points monotonically progressing along the line connecting $v_0$ and $r$; hence it is convergent and $\Delta v_n=v_{n}-v_{n-1} \to 0$ as $n \to \infty$.   Defining $\Delta\hat x_n = \hat x_n - H^x_\infty v_n$ and using (\ref{equ:6}) and
$A_{\tt cl} H_\infty^x+B_{\tt cl} = H_\infty^x$,
leads to $\Delta \hat x_{n+1} = A_{\tt cl}\Delta \hat x_n + H^x_\infty \Delta v_n + B_w e_n$. Note that 
$e_n \in \Omega_{{n}} \subseteq \Omega_{\bar{n}}$
for $n \geq \bar{n}$.
Define $\bar{F}=\bigoplus_{i=0}^\infty A_{\tt cl}^i B_w
{\Omega}_{\bar{n}}$ and note that $\bar{F}$ is compact as $\Omega_{\bar{n}}$ is compact and $A_{\tt cl}$ is Schur.  Since $\Delta v_n \to 0$
and $A_{\tt cl}$ is Schur, it follows that $\Delta \hat x_{n} \to \bar{F}$  and $ \hat x_n \to 
\{H^x_\infty v_n\} \bigoplus \bar{F}$ as $n \to \infty$.

Note that $\bar{O}_\infty \subseteq \tilde{O}_{\infty,n}$ for $n \geq \bar{n}$. Additionally,
 for sufficiently small $\varepsilon>0$ it follows that \cite[Theorem 7.1]{kolmanovsky1998theory}:
% for systems with the constant in time disturbance set ${\Omega}_{\bar{n}}$, for sufficiently small $\varepsilon>0$ and $n>\bar{n}$ it follows that \cite[Theorem 7.1]{kolmanovsky1998theory}:
\begin{align*}
&L_\infty \subseteq \bar{O}_\infty,\text{where}\\
&L_\infty=\left\{(v,\hat{x}):~H_\infty^y v \in \tilde{Y}_{\bar{n}}, \hat{x} \in \{H_\infty^x v\} \oplus \bar{F} \oplus \varepsilon \mathcal{B}^{n_x} \right\}. 
\end{align*}

Given $\varepsilon>0$ such that $L_\infty \subseteq \bar{O}_\infty \subseteq \tilde{O}_{\infty,n}$, the convergence 
 $\hat x_n \to 
\{H^x_\infty v_n\} \oplus \bar{F}$ as $n \to \infty$ implies that for any sufficiently large $n \geq \bar{n}$,
$$\hat{x}_{n-1} \in \{ H_\infty^x v_{n-1} \} 
\oplus \bar{F} \oplus 
\frac{\varepsilon}{2 \|A_{\tt cl}\|}.$$
Then, using $A_{\tt cl} H_\infty^x+B_{\tt cl}=H_\infty^x$,  $A_{\tt cl} \bar{F} +B_w {\Omega}_{\bar{n}} \subseteq \bar{F}$ yields 
\begin{eqnarray*}
\hat{x}_n & = & A_{\tt cl} \hat{x}_{n-1}+B_{\tt cl} v_{n-1} +B_w e_{n-1} \\
& \in & (A_{\tt cl} H_\infty^x + B_{\tt cl}) v_{n-1} \oplus (A_{\tt cl} \bar{F} +B_w 
{\Omega}_{\bar{n}}) \oplus \frac{\varepsilon}{2} {\mathcal B}^{n_x} 
\\
&  \subseteq & 
\{ H_\infty^x v_{n-1} \} \oplus \bar{F} \oplus \frac{\varepsilon}{2}  {\mathcal B}^{n_x}.
\end{eqnarray*}
Hence, for any $v_n$ such that \begin{equation}\label{equ:choice}\|v_n-v_{n-1}\|\leq \delta= \frac{\varepsilon}{2 \|H_\infty^x\|} \end{equation} 
it follows that $\hat{x}_n \in  \{H_\infty^x v_{n}\} \oplus \bar{F} \oplus {\varepsilon}  {\mathcal B}^{n_x}. $

Since $H_\infty^y v_{n-1} \in \tilde{Y}_{\bar{n}}$,
$H_\infty^y r \in \tilde{Y}_{\bar{n}}$,
and $\tilde{Y}_{\bar{n}}$ is convex, for any such $v_n=v_{n-1}(1-\kappa)+r \kappa$, $0 \leq \kappa \leq 1$, 
$H_\infty^y v_n \in \tilde{Y}_{\bar{n}}$.
Thus any choice of $v_n$ satisfying (\ref{equ:choice}) and $v_n=v_{n-1}(1-\kappa)+r \kappa$, $0 \leq \kappa \leq 1$,
results in $(v_n,\hat{x}_n) \in \tilde{O}_{\infty,n}$ and is feasible.  Since $v_n$ is chosen using maximization of $\kappa_n$ in (\ref{eq:kappaOpt}) the possibility of taking bounded from below steps in $v_n$ for arbitrary large $n$
contradicts $v_n - v_{n-1} \to 0$ unless $v_n=r$ for some finite $n$.
\end{proof}

\section{Unmatched Disturbances}
\label{sec:unmatchedDist}
In this section we describe an approach to handle  a class of constrained systems where the control input and disturbance input are not matched, with the model given by
\begin{equation}\label{equ:ud1}
	x_{n+1} = Ax_n + Bu_n + B_1 w_n,
	\end{equation}
where $n_u$ might differ from $n_w$. Suppose the constraints are imposed on a vector of chosen outputs of the system,
\begin{equation}\label{equ:ud1.5}
	z_n = Hx_n \in Z,
\end{equation}
and assume that $range\{HB_1\}\subseteq \;range\{HB\}$ while the matrix $HB$ has full column rank so that  $(HB (HB)^\top)$ is invertible.  Define
\begin{equation}\label{equ:ud2}
	(HB)^\dagger=(HB)^\top (HB (HB)^\top)^{-1},
\end{equation}
so that $(HB) (HB)^\dagger=I$ and the control
\begin{equation}\label{equ:ud3}
	u_n=(HB)^\dagger\left(-HA \hat{x}_n-HB_1 \hat{w}_n + \Lambda \hat{z}_n +\Gamma v_n \right),
\end{equation}
where $\Lambda$ is a Schur matrix.
Based on the observer (\ref{eq:input_MatchedDist.1}), 
\begin{equation*}
	\hat{z}_{n+1}  = H \hat{x}_{n+1} = HA \hat{x}_n + HB u_n + HB_1 \hat{w}_n+ HL_xC_z e_{x,n}. 
 \end{equation*}
After replacing in (\ref{equ:ud3}), we obtain
\begin{align}
	\hat{z}_{n+1}& = \Lambda \hat{z}_n+\left[\begin{array}{cc} H L_x C & 0 \end{array} \right]e_n+\Gamma v_n,	\\
	z &= \hat{z}_n +\left[\begin{array}{cc} H & 0_{n_z\times n_w} \end{array} \right]e_n \in Z.	
\end{align}
This fits the model (\ref{eq:predModFin}) with $\hat{x} = \hat{z}$,  $y_{\tt cl}=z$,
$A_{\tt cl}=\Lambda$, $B_{\tt cl}=\Gamma$,
$B_w=\left[\begin{array}{cc} H L_1 C & 0 \end{array} \right]$,
$D_w=\left[\begin{array}{cc} H & 0 \end{array} \right]$, $C_{\tt cl}=I_{n_z \times n_z}$, $D_{\tt cl}=0$
and $Y_{\tt cl}=Z$. Note that 
in the Luenberger observer, the matrix
 $A_{\tt ext}$ matrix needs to be modified as
\begin{equation}
    A_{\tt ext} =  \begin{bmatrix}A &  B_1\\0_{n_w\times n_x}&I_{n_w}\end{bmatrix}
    \label{eq:A_obs_unmatched}
\end{equation}
\begin{remark}
For the particular case where $range\{B_1\}\subseteq range\{B\}$ and $B$ has full column rank one can consider the full dynamics as well as constraints on inputs and all states by applying the developments of section \ref{sec:matchedDist}, considering $-B^\dagger B_1\hat w_n$ as the disturbance cancellation term in the control law (\ref{eq:input_MatchedDist}) and defining $A_{\tt ext}$ following equation (\ref{eq:A_obs_unmatched}).
\end{remark}

\section{Simulation results}

\subsection{ Comparison of different error bounding schemes}

In order to illustrate how the choice of the error bounding affects the performance of the presented robust reference governor we first consider a mass-spring-damper system,
\begin{equation} m \ddot{x}_1+c_d \dot{x}_1+k x_1=(u+w),\end{equation}
where $x_1$ denotes the position, $x_2=\dot{x}$ denotes the velocity,
$u$ is the control force and $w$ is the disturbance force.
  In this example we consider that the open-loop system is lightly damped with $k=1.4400$, $c_d=0.24$ and $m=1$.
The system model is converted to discrete time using a sampling period of $T_s=0.2$ sec.  We consider that only the position of the mass can be measured, therefore $C=[1,~0]$.
The Luenberger observer is designed with $L=\begin{bmatrix}
    0.3796 & 0.2524 & 0.2491
\end{bmatrix}$. 
We assume that $E_0=\{0\} \times \{0\} \times [-0.4,\;0.4]$, and $\tilde{W}= [-0.001,~0.001]$.  Constraints are imposed on the position as $|x_1| \leq 1.2$, while control constraints are not considered. An LQR controller has been designed leading to $K= [-0.3962,~ -0.8888]$ and a reference tracking feedforward gain 
$\Gamma = 1.8362$.

We next compare the polyhedral and ellipsoidal approaches for the error bounding. For the polyhedral bounding $n_L=100$ normalized directions were chosen at random and $\Omega_n$ was computed for $n=0,1,\cdots, 400$; as it appeared that $\Omega_n \subseteq \Omega_{200}$ for $n>200$, we set $\Omega_{\tt f} = \Omega_{200}$ and $\Omega_n=\Omega_{\tt f}$, $n\geq 200$.
Similarly, the ellipsoidal bounds were computed for $n=0,\cdots,200$ and $\Omega_{\tt f}=\Omega_{200}$ was chosen. Note that for these disturbance bounds, error bounds, system and observer choice, the sequence of $\xi_k$ defined in (\ref{eq:GammaDef}) is decreasing and $\xi_{200}\approx \frac{\rho}{1-\mu}$. 

Additionally, we define $\tilde Y_{n} = Y_{\tt cl }\ominus\mathcal E^y(n,200)\ominus \epsilon\mathcal B^{ny}$ where $\epsilon = 0.01$.
The sequence $\tilde{O}_{\infty,n}$ is then constructed for $n=0,1,\cdots,200$ and by Proposition \ref{prop:MOAS_cst1}, $\tilde{O}_{\infty,n}=\tilde{O}_{\infty,200}$ for $n > 200$. % Every $\tilde{O}_{\infty,n}$ in the sequence had $k^*(n)<45$.

As shown in Figure~\ref{fig:damper_pos}, irrespective of the choice of the error bounding sets, the constraints are strictly enforced at all times. Nevertheless, the ellipsoidal  bounding leads to a more conservative response as the reference command is responding much slower and reaches lower values over the first $10$ seconds when compared with the polyhedral error bounding case. Similarly, we observe that the final error bound is larger for the ellipsoidal bounding than the polyhedral bounding as seen by the reference command values $v_n$ never reaching $r_n$.

\begin{figure}[htb]
	\begin{center}
		\includegraphics[scale=0.6]{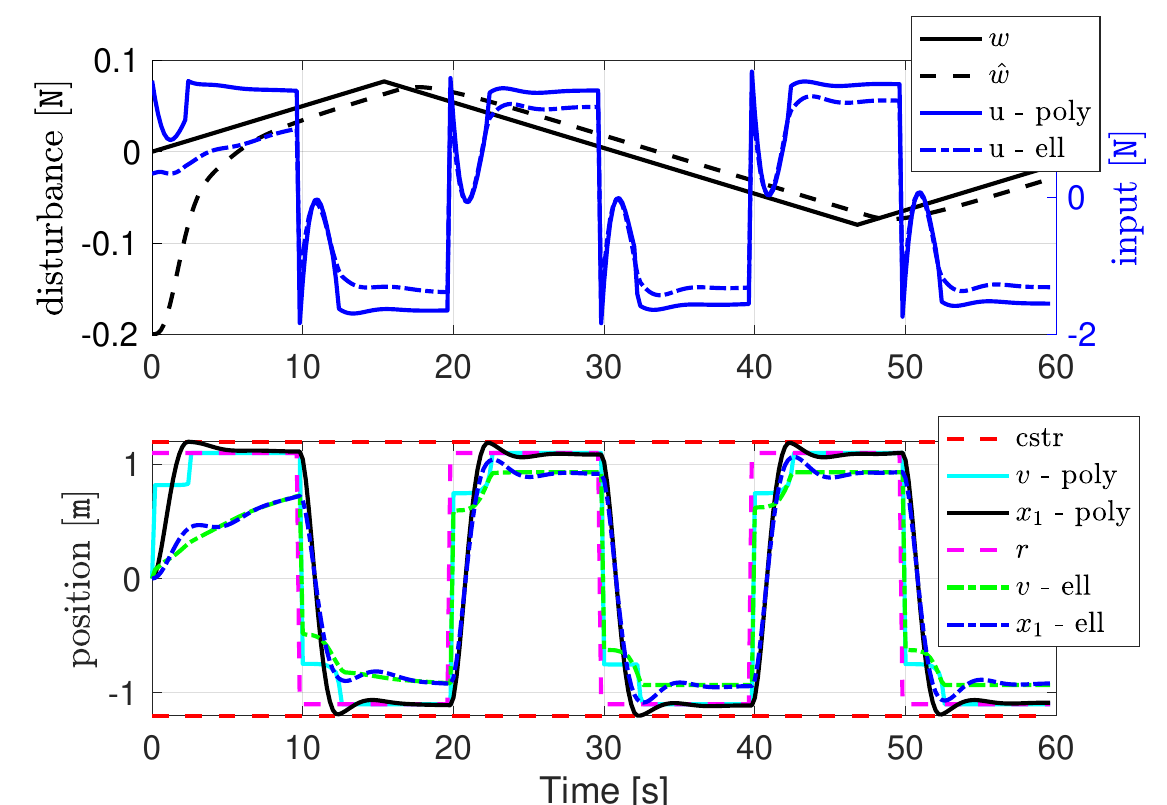} %scale was 0.4 for 2 columns
		\caption{Top: Time histories of the actual and estimated disturbance and of the control inputs 
         for the polyhedral and ellipsoidal bounding-based RGs. Bottom:  Time histories of the actual and estimated position, desired reference and modified by RG reference for the polyhedral and ellipsoidal bounding based RGs.}
		\label{fig:damper_pos}
	\end{center}
\end{figure}

The difference between the two bounding methods is evident in Figure \ref{fig:damper_sets} where we observe that $\Omega_{200}$ obtained using the polyhedral approach is much smaller than the one obtained using the ellipsoidal approach. It is worth mentioning that even though the ellipsoidal bounding approach is more conservative compared to the polyhedral one, it is computationally less demanding and produces a family of nested error bounding sets.

\begin{figure}[htb]
	\begin{center}
		\includegraphics[scale=0.8]{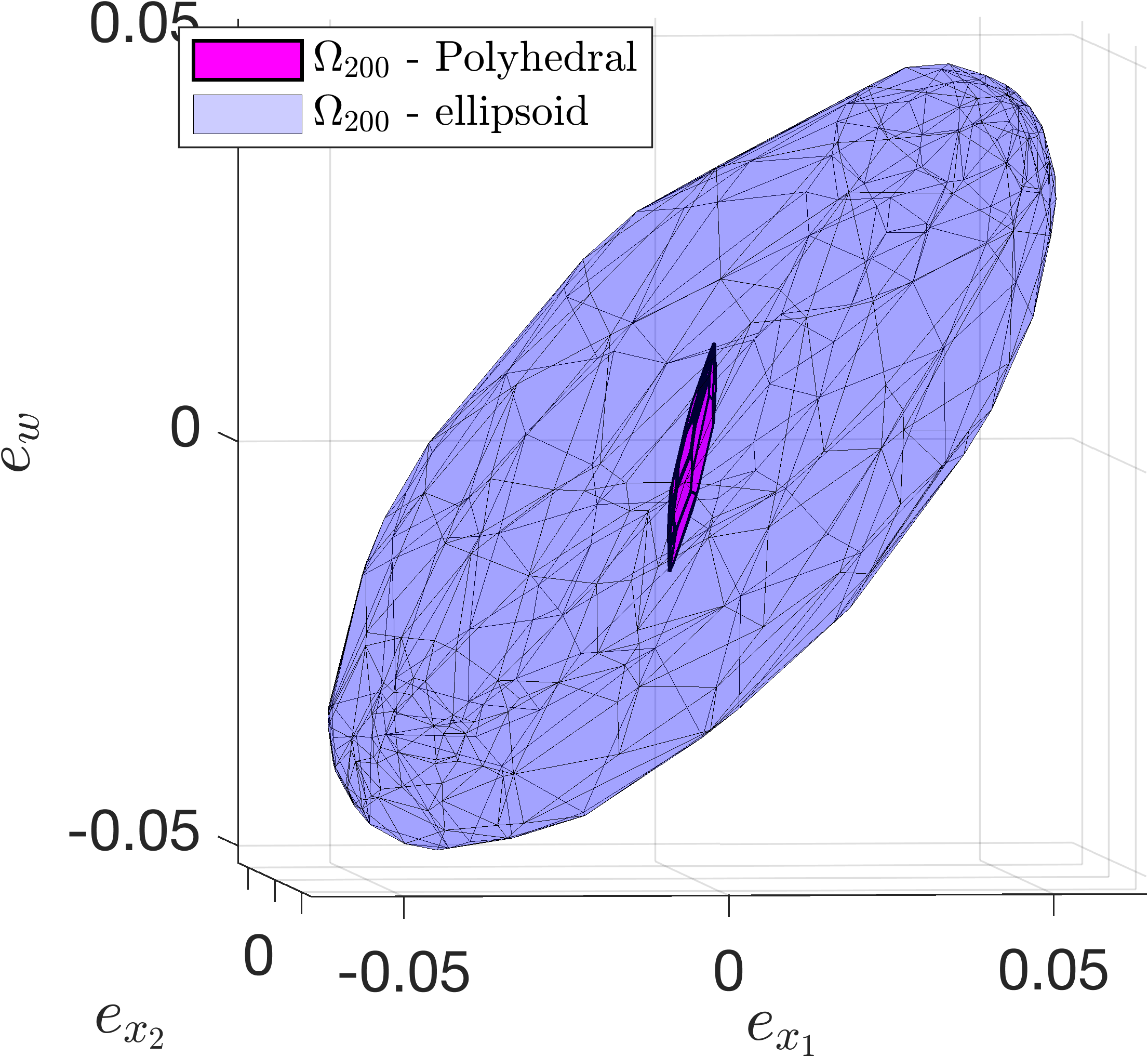}%scale was 0.6 for 2 columns
		\caption{Error bounding sets $\Omega_{200}$ obtained using ellipsoidal and polyhedral error bounding.}
		\label{fig:damper_sets}
	\end{center}
\end{figure}

\subsection{Short period dynamics of transport aircraft}

We now compare the proposed RG to a more classical robust RG such as the one described in \cite[Section 2.2]{garone2017reference}. In this section we will refer to the former as case a) and to the latter as case b). For the comparison, we consider the short period dynamics of a generic aircraft flying at 0.8 $\tt Ma$ and 30000 ${\tt ft}$ modeled by \cite{poloni2014disturbance}
\begin{equation}
    \begin{bmatrix}
    \dot \alpha\\\dot \theta\\ \dot q
    \end{bmatrix} = \begin{bmatrix}
        -0.7018& 0 &0.9761\\
        0 & 0 & 1\\
        -2.6923& 0 &-0.7322
    \end{bmatrix}\begin{bmatrix}
        \alpha\\ \theta\\ q
    \end{bmatrix} + \begin{bmatrix}
        -0.0573\\0\\-3.5352
    \end{bmatrix}(\delta_e + w)
\end{equation}
where $\alpha\units{rad}$ is the angle of attack (AoA), $\theta\units{rad}$  the pitch angle, $q\units{rad\; s^{-1}}$ is the pitch rate, $\delta_e \units{rad}$ is the elevator deflection and $w$ is a matched disturbance input. The aircraft is to follow a pitch angle reference command.

The system is subject to constraints on both states and the sum of input and disturbance specified as follows: $|\alpha|\leq 2^\circ$ and $|\delta_e|\leq 40^\circ$. The model is converted to discrete time using a sampling period $T_s = 20{\tt  ms}$. 

The AoA and pitch angle are assumed to be measured and a Luenberger observer is designed with
\begin{equation*}
    L = \begin{bmatrix}
        0.1752 &    0.0201    &    0.2121&   -0.1184 \\
        0.0201  & 0.1881 &    0.2431&  -0.1123
\end{bmatrix}^\top .
\end{equation*}

The disturbance, in $\tt rad$, and its rate of change are bounded so that
$W=[-0.45,0.45]$ and $\tilde{W}=[-0.0004,0.0004]$.  The initial estimation errors of the states and of the disturbance are assumed to be bounded in $E_0 = [-10^{-10},10^{-10}] \times$ $[-10^{-10},10^{-10}] \times$ 
 $[-10^{-5},10^{-5}] \times$ $ [-10^{-2},10^{-2}],$
where error values relating to angles are specified in radians. 

A stabilizing LQR controller (\ref{eq:input_MatchedDist}) is designed with $$K =\begin{bmatrix} 4.0441 & 103.2201 & 7.5627\end{bmatrix}$$ and the feedforward gain $\Gamma = -103.22$. 

In case a), the error bounds are computed using the polyhedral 
bounding approach, with $n_L = 100$ directions % $\{L_j\}_{j=1}^{n_L}$ used in (\ref{eq:OmegaPoly}) 
selected at random in $\mathbb R^{n_x+n_w}$. 
A ``terminal'' error bounding set $\Omega_{\tt f} \supseteq \bigcup_{n=200}^{250}\Omega_k$ was chosen for which the set inclusion, $\Omega_n \subseteq\Omega_{\tt f}$, $n=200,\cdots, 600,$ was verified.  Then for $n \geq 200$, $\Omega_n$ was set to $\Omega_{\tt f}$. 
A vertex-representation of $\Omega_{\tt f}$ was obtained by selecting all vertices of the sets $\Omega_n$, $n =200,\cdots,250$ and deleting the ones strictly inside the convex hull of all vertices. Sets $\tilde Y_{n}$ are defined as $\tilde Y_{n} = Y_{\tt cl }\ominus\mathcal E^y(n,400)\ominus \epsilon\mathcal B^{ny}$ where $\epsilon = 0.01$.
The sequence $\tilde{O}_{\infty,n}$ was constructed for $n=0,1,\cdots,200$ and, by Proposition \ref{prop:MOAS_cst1},
$\tilde{O}_{\infty,n}=\tilde{O}_{\infty,200}$ for $n > 200$.
% Every $\tilde{O}_{\infty,n}$ in the sequence had $k^*(n)<30$.

Figures \ref{fig:shrtDyns_Rtrack}, \ref{fig:shrtDyns_uHist} show the time histories during a pitch angle reference tracking with RG in both cases a) and  b). In both cases, the RG ensures that constraints are respected even in the presence of disturbances. Without RG, constraints on both input and state were violated. In this example, it is clear that the proposed RG in case a) is less conservative that the conventional robust RG in case b), as shown by the slower responding reference command and the fact that states and inputs are further away from the constraint boundary.

\begin{figure}[ht]
\centering
\includegraphics[scale=0.65]{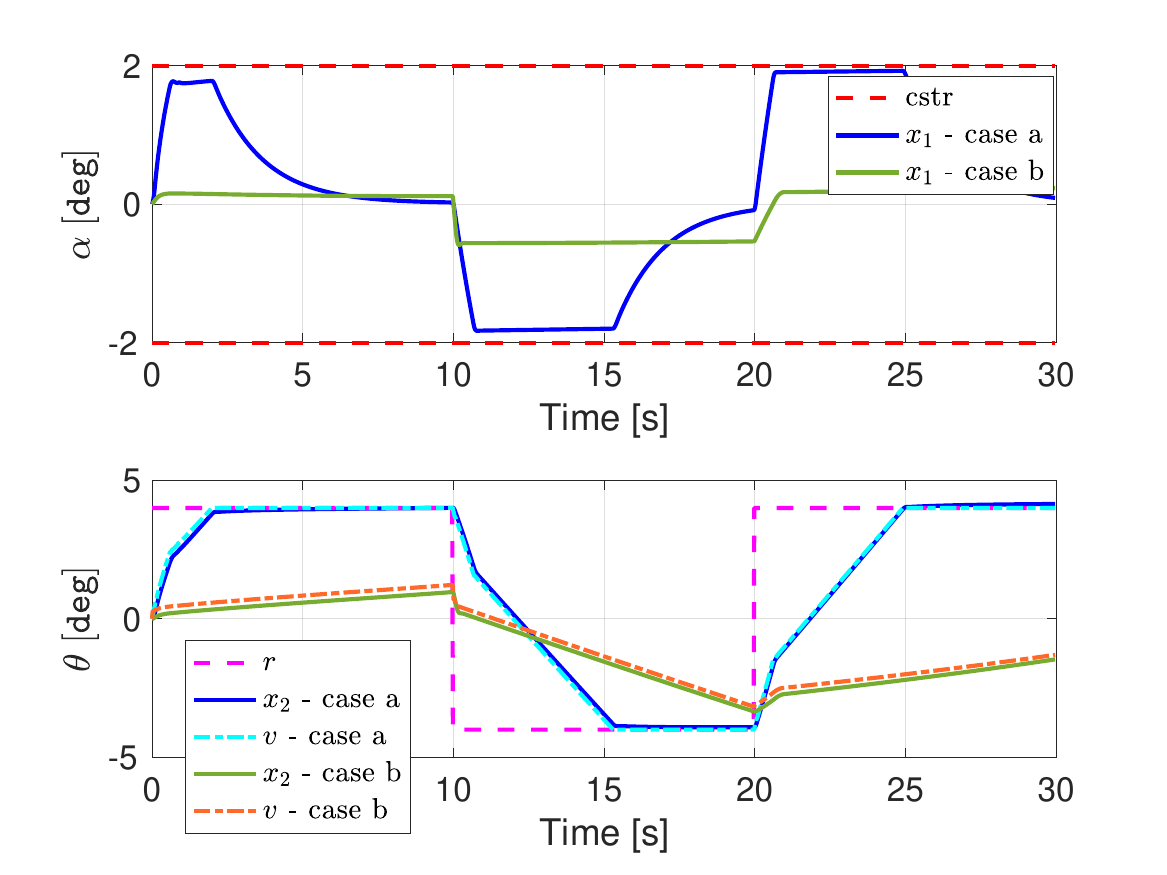}%scale was 0.45 for 2 columns
\caption{Time histories of states, constraints,  desired reference and referencemodified by RG reference for the aircraft short period dynamics with RG based on disturbance estimation (case a) and conventional robust RG (case b). }
\label{fig:shrtDyns_Rtrack}
\end{figure}

\begin{figure}[ht]
\centering
\includegraphics[scale=0.65]{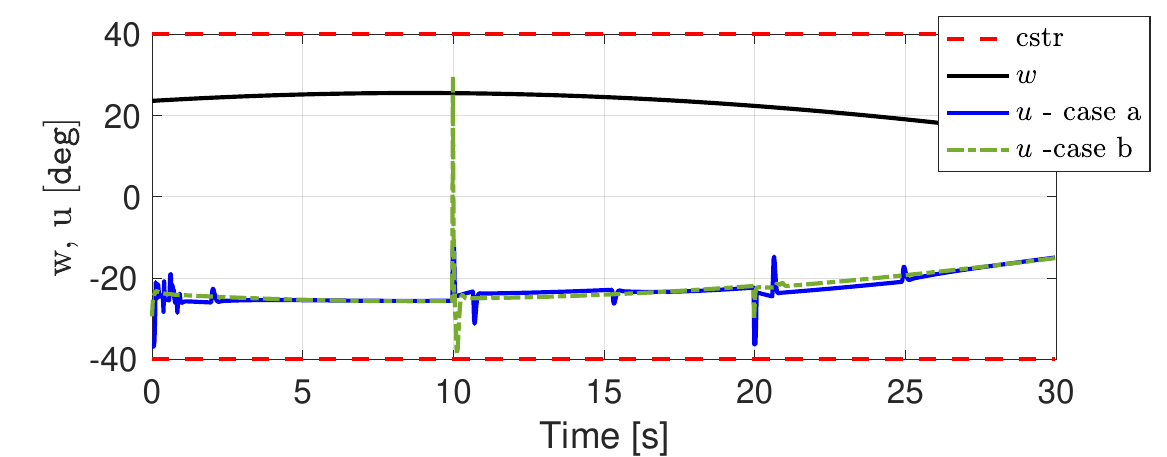}%scale was 0.45 for 2 columns
\caption{Time histories of the input, disturbance and constraints for the aircraft short period dynamics with RG making use of disturbance estimation (case a) and conventional robust RG (case b). }
\label{fig:shrtDyns_uHist}
\end{figure}

Figure \ref{fig:shrtDyns_OmegaPoly} shows projections of several instances of the error bounding sets $\Omega_n$ illustrating that the inclusion $\Omega_{n+1}\subseteq\Omega_n$ does not hold, in general. Moreover,  the error bounding sets are not strictly decreasing in volume either.
This occurs due to the propagation of the initial error (primarily in the disturbance estimate) to other state estimates due to the error dynamics.

\begin{figure}[ht]
\centering
\includegraphics[scale=0.8]{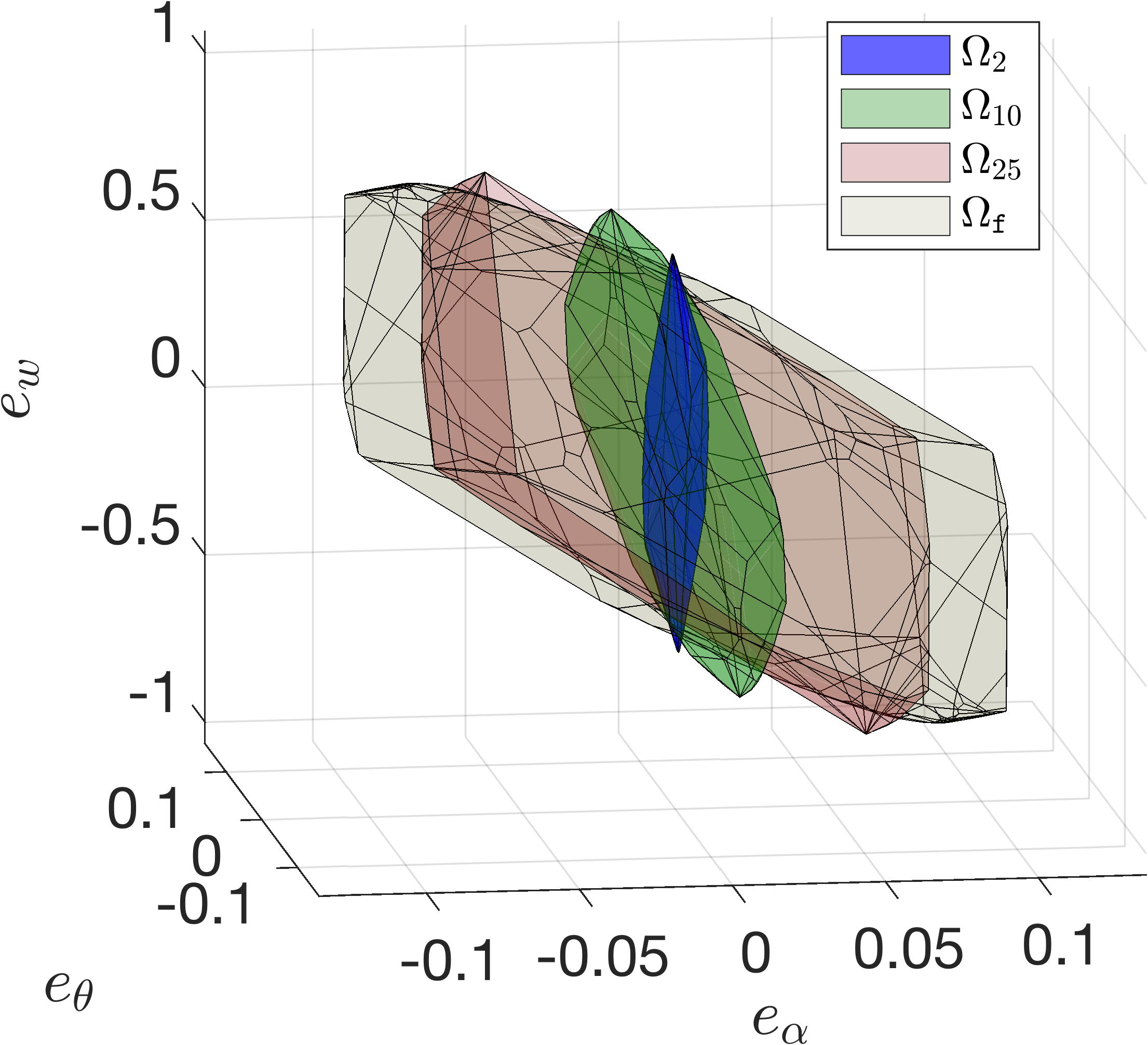}%scale was 0.6 for 2 columns
\caption{Projections of several error bounding sets $\Omega_n$ onto the $(e_\alpha,e_\theta,e_w)$ error space.}
\label{fig:shrtDyns_OmegaPoly}
\end{figure}

\subsection{Aircraft under icing conditions }
As an illustration of the case when the disturbance and control inputs are not matched (Section \ref{sec:unmatchedDist}), we consider the linearized longitudinal dynamics for a generic transport aircraft at 168 ${\tt  fps}$ and $8000 {\tt  ft}$ affected by ice, modeled by \cite{poloni2014disturbance}: 
\begin{align*}
    \begin{bmatrix}
        \dot {\tt v}\\ \dot q\\ \dot \alpha \\ \dot \theta
    \end{bmatrix}&= \begin{bmatrix}
       -0.042& -0.023& -32.22& -32.20\\
    0.015 & -3.80 & -52.82 & 0\\
    -0.001 & 0.961 & -2.990 & 0\\
    0 & 1 & 0 & 0
        \end{bmatrix} \begin{bmatrix}
         {\tt v}\\  q\\  \alpha \\  \theta
    \end{bmatrix} +\\
    & \begin{bmatrix}
     -0.052 & 0.177\\
    -1.035 & 0.018\\
    -0.005 & 0\\
    0 &0
    \end{bmatrix}\begin{bmatrix}
        \delta_e\\ \delta_t
    \end{bmatrix} + \begin{bmatrix}
        -100.5& 5.27\\
    0 &0 \\
    -0.031& -0.599 \\
    0& 0 
    \end{bmatrix}\begin{bmatrix}
        \delta_{C_D}\\ \delta_{C_L}
    \end{bmatrix}.
\end{align*}

All states are deviations from nominal, in particular, deviations from the longitudinal speed, pitch rate, angle of attack and pitch angle. Inputs represent deviations from nominal elevator deflection and thrust respectively. The disturbances, unmatched, represent the effect of icing on drag ($w_1=\delta_{C_D}$) and lift ($w_2=\delta_{C_L}$) coefficients, respectively. We consider constraints: $|{\tt v}|\leq 5$ ${\tt  fps}$ and $|q|\leq 12$ ${\tt  deg\;s^{-1}}$ and as such consider these two states as outputs. Note that $range(B_1) = range(B)$ and thus the condition $range(HB_1) = range(HB)$ is fulfilled. The model is converted to discrete-time with a sampling period $T_s = 30$ ${\tt  ms}$, and ${\tt v}$ and $\alpha$ are assumed to be measured outputs.  The observer gain is chosen as
\begin{equation*}
    L = \begin{bmatrix}
       0.764 &  -0.036   & 0.0005 &  -0.008  & -0.15  &  0.03\\
   -0.008   & 0.49  & -0.04   & 0.16   & 0.01    &0.91
    \end{bmatrix}^\top\!\!.
\end{equation*}

The rate of change of the disturbance is bounded and  $\tilde{W}=[-2.5 \times 10^{-5},2.5 \times 10^{-5}] 
\times 
[-2.5 \times 10^{-4},  
2.5 \times 10^{-4}]$.  The initial values of the state and disturbance estimation errors were assumed to lie in the set $E_0 \in [-1,1] \times [-10^{-2},10^{-2}] \times [-10^{-4},10^{-4}] \times [-10^{-4},10^{-4}] \times
[-10^{-3},10^{-3}] \times [-10^{-3},10^{-3}].$

Additionally the following values were selected: 
\begin{equation}
    \Gamma = \begin{bmatrix}
          0.0208 &   0.0126\\
   -0.0030   & 0.0727
    \end{bmatrix},\quad \Lambda = \begin{bmatrix}
            0.9792  & -0.0126\\
    0.0030   & 0.9273
    \end{bmatrix}.
\end{equation}

To construct $\Omega_n$, the polyhedral bounding approach was used with  $n_L=100$ randomly selected directions. The terminal error bounding set $\Omega_{\tt f} \supseteq \bigcup_{n=60}^{70}\Omega_n$ was chosen and the set inclusion $\Omega_n \subseteq\Omega_{\tt f}$, $n=60,\cdots,150$ was verified. Then the set $\Omega_{\tt f}$ was used in place of  $\Omega_n$ for $n \geq 60$.  A vertex representation of $\Omega_{\tt f}$ was obtained by selecting all vertices of the sets $\Omega_n$, for 
$n =60,\cdots,70$ and deleting the ones strictly inside the convex hull of all vertices.
Additionally, $\tilde Y_{n} = Y_{\tt cl }\ominus\mathcal E^y(n,120)\ominus \epsilon\mathcal B^{ny}$ with $\epsilon = 0.01$.
Finally, the sequence $\tilde{O}_{\infty,n}$ was constructed, where following Proposition \ref{prop:MOAS_cst1},
$\tilde{O}_{\infty,n}=\tilde{O}_{\infty,60}$ for $n\geq 60$.
% Every $\tilde{O}_{\infty,n}$ in the sequence had $k^*(n)<120$.

Figure \ref{fig:cc_states} shows time histories of the disturbances and outputs. Although the desired reference command is outside of the constraint bounds and is therefore never reached, the applied reference command reaches a value very close to the constraint boundary. Note that, although the disturbances are substantial (drag increase $\approx 50\%$ , lift reduction  $\approx 100\%$) the cancellation is effectively performed in the outputs. Other states, however, are significantly affected (not shown).

\begin{figure}[ht]
\centering
\includegraphics[scale=0.65]{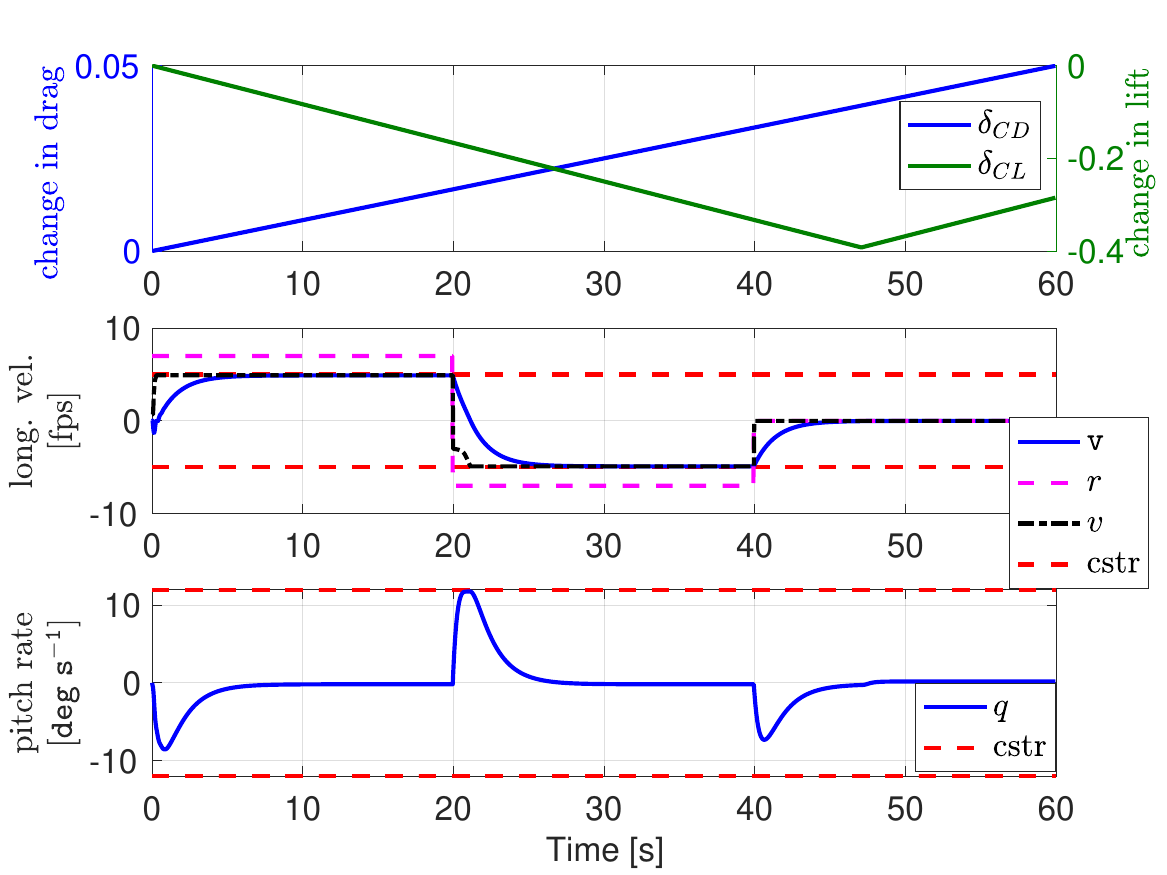}%scale was 0.45 for 2 columns
\caption{Time histories of (estimated and actual) disturbances, outputs, desired reference, RG modified reference, and constraints.}
\label{fig:cc_states}
\end{figure}
\section{Conclusion}
In this paper we considered the application of reference governors (RG) to a class of linear, discrete time systems with unmeasured disturbance limits and output measurements only. We exploited disturbance input and state estimation to reduce the conservatism of the reference governor operation.
To accommodate time dependent bounds on the disturbance and state estimation errors we have employed a time dependent sequence of constraint admissible sets for reference governor implementation.  We have shown that reference governor properties, such as finite time convergence to strictly admissible in steady-state constant inputs, are preserved by our reference governor implementation. Numerical examples illustrated the potential for the use of such a reference governor in aerospace applications, the reduction in conservatism and the advantages of the polyhedral bounding over ellipsoidal bounding.

% \bibliographystyle{plain}
% \bibliography{biblio.bib}

\begin{thebibliography}{10}

\bibitem{ANGELI2001151}
D.~Angeli, A.~Casavola, and E.~Mosca.
\newblock On feasible set-membership state estimators in constrained command governor control.
\newblock {\em Automatica}, 37(1):151--156, 2001.

\bibitem{chen2015disturbance}
Wen-Hua Chen, Jun Yang, Lei Guo, and Shihua Li.
\newblock Disturbance-observer-based control and related methods—an overview.
\newblock {\em IEEE Transactions on industrial electronics}, 63(2):1083--1095, 2015.

\bibitem{garone2017reference}
Emanuele Garone, Stefano Di~Cairano, and Ilya Kolmanovsky.
\newblock Reference and command governors for systems with constraints: A survey on theory and applications.
\newblock {\em Automatica}, 75:306--328, 2017.

\bibitem{gilbert1991linear}
Elmer~G Gilbert and K~Tin Tan.
\newblock Linear systems with state and control constraints: The theory and application of maximal output admissible sets.
\newblock {\em IEEE Transactions on Automatic control}, 36(9):1008--1020, 1991.

\bibitem{kolmanovsky1995maximal}
Ilya Kolmanovsky and Elmer Gilbert.
\newblock Maximal output admissible sets for discrete-time systems with disturbance inputs.
\newblock In {\em Proceedings of 1995 American Control Conference-ACC'95}, volume~3. IEEE, 1995.

\bibitem{kolmanovsky1998theory}
Ilya Kolmanovsky and Elmer~G Gilbert.
\newblock Theory and computation of disturbance invariant sets for discrete-time linear systems.
\newblock {\em Mathematical problems in engineering}, 4(4):317--367, 1998.

\bibitem{mayne2005robust}
David~Q Mayne, Mar{\'\i}a~M Seron, and SV~Rakovi{\'c}.
\newblock Robust model predictive control of constrained linear systems with bounded disturbances.
\newblock {\em Automatica}, 41(2):219--224, 2005.

\bibitem{MAYNE20061217}
D.Q. Mayne, S.V. Raković, R.~Findeisen, and F.~Allgöwer.
\newblock Robust output feedback model predictive control of constrained linear systems.
\newblock {\em Automatica}, 42(7):1217--1222, 2006.

\bibitem{MAYNE20092082}
D.Q. Mayne, S.V. Raković, R.~Findeisen, and F.~Allgöwer.
\newblock Robust output feedback model predictive control of constrained linear systems: Time varying case.
\newblock {\em Automatica}, 45(9):2082--2087, 2009.

\bibitem{ong2006minimal}
Chong-Jin Ong and Elmer~G Gilbert.
\newblock The minimal disturbance invariant set: Outer approximations via its partial sums.
\newblock {\em Automatica}, 42(9):1563--1568, 2006.

\bibitem{ossareh2020reference}
Hamid~R Ossareh.
\newblock Reference governors and maximal output admissible sets for linear periodic systems.
\newblock {\em International Journal of Control}, 93(1):113--125, 2020.

\bibitem{poloni2014disturbance}
Tom{\'a}{\v{s}} Pol{\'o}ni, Uro{\v{s}} Kalabi{\'c}, Kevin McDonough, and Ilya Kolmanovsky.
\newblock Disturbance canceling control based on simple input observers with constraint enforcement for aerospace applications.
\newblock In {\em 2014 IEEE Conference on Control Applications (CCA)}, pages 158--165. IEEE, 2014.

\bibitem{rakovic2005invariant}
Sasa~V Raković, Eric~C Kerrigan, Konstantinos~I Kouramas, and David~Q Mayne.
\newblock Invariant approximations of the minimal robust positively invariant set.
\newblock {\em IEEE Transactions on automatic control}, 50(3):406--410, 2005.

\bibitem{SHINGIN2004389}
Hidenori Shingin and Yoshito Ohta.
\newblock Optimal invariant sets for discrete-time systems: Approximation of reachable sets for bounded inputs.
\newblock {\em IFAC Proceedings Volumes}, 37(11):389--394, 2004.
\newblock 10th IFAC/IFORS/IMACS/IFIP Symposium on Large Scale Systems 2004.

\end{thebibliography}

\end{document}